\documentclass[a4paper, 11pt]{article}
\usepackage[usenames]{xcolor}
\usepackage{xifthen}

\def\fontsettingup{2} 

\usepackage{amsmath, amsthm, amssymb}
\usepackage[T1]{fontenc}
\ifthenelse{\fontsettingup = 1}{ \usepackage{eulerpx, eucal, tgpagella}   }{}
\ifthenelse{\fontsettingup = 2}{ \usepackage{mathpazo, tgpagella} }{}

\usepackage{hyperref}
\hypersetup{
  colorlinks=true,
  citecolor=blue,
  linkcolor=blue,
  filecolor=magenta,
  urlcolor=cyan,
}

\usepackage[linesnumbered,boxed,ruled,vlined]{algorithm2e}
\usepackage{tikz}
\usepackage{cleveref}

\usepackage[a4paper]{geometry}
\newgeometry{
textheight=9in,
textwidth=6.5in,
top=1in,
headheight=14pt,
headsep=25pt,
footskip=30pt
}

\newtheorem{theorem}{Theorem}

\newtheorem*{claim*}{Claim}

\newtheorem{fact}[theorem]{Fact}
\newtheorem{lemma}[theorem]{Lemma}
\newtheorem{proposition}[theorem]{Proposition}
\newtheorem{corollary}[theorem]{Corollary}
\theoremstyle{definition}

\newtheorem{definition}[theorem]{Definition}
\newtheorem{remark}[theorem]{Remark}
\newtheorem*{remark*}{Remark}

\ifthenelse{\fontsettingup = 1}{
  \def\*#1{\mathbf{#1}} 
  \def\+#1{\mathcal{#1}} 
  \def\-#1{\mathrm{#1}} 
  \def\^#1{\mathbb{#1}} 
  \def\!#1{\mathfrak{#1}} 
}{}

\ifthenelse{\fontsettingup = 2}{
  \def\*#1{\boldsymbol{#1}} 
  \def\+#1{\mathcal{#1}} 
  \def\-#1{\mathrm{#1}} 
  \def\^#1{\mathbb{#1}} 
  \def\!#1{\mathfrak{#1}} 
}{}

\def\oPr{\mathbf{Pr}}
\renewcommand{\Pr}[2][]{ \ifthenelse{\isempty{#1}}
  {\oPr\left[#2\right]}
  {\oPr_{#1}\left[#2\right]} } 

\def\oE{\mathbb{E}}
\newcommand{\E}[2][]{ \ifthenelse{\isempty{#1}}
  {\oE\left[#2\right]}
  {\oE_{#1}\left[#2\right]} }

\DeclareMathOperator*{\oVar}{\mathbf{Var}}
\newcommand{\Var}[2][]{ \ifthenelse{\isempty{#1}}
  {\oVar\left[#2\right]}
  {\oVar_{#1}\left[#2\right]} }

\def\oEnt{\mathbf{Ent}}
\newcommand{\Ent}[2][]{ \ifthenelse{\isempty{#1}}
  {\oEnt\left[#2\right]}
  {\oEnt_{#1}\left[#2\right]} }

\newcommand{\DTV}[2]{\-D_{\mathrm{TV}}\left({#1},{#2}\right)}

\renewcommand{\epsilon}{\varepsilon}
\renewcommand{\emptyset}{\varnothing}

\newcommand{\set}[1]{\left\{#1\right\}}
\newcommand{\tuple}[1]{\left(#1\right)} \newcommand{\eps}{\varepsilon}

\newcommand{\tp}{\tuple}

\newcommand{\abs}[1]{\left\vert#1\right\vert}

\newcommand{\defeq}{:=}
\newcommand{\numP}{\#{\textnormal{\textbf{P}}}}

\newcommand{\Zrel}[1][G,\*p]{Z_{\-{rel}}(#1)}

\newcommand{\murel}[1][G,\*p]{\mu^{\-{NR}}_{#1}}
\newcommand{\muM}[1][\+M,\*\lambda]{\mu_{#1}}
\newcommand{\muB}[1][\+B,\*\lambda]{\mu_{#1}}

\DeclareMathOperator{\rk}{rk}
\DeclareMathOperator{\wt}{wt}
\newcommand{\RCpolar}{\widehat{\pi}_{RC,q}}

\title{Near-linear time samplers for matroid independent sets with applications}
\date{}
\author{Xiaoyu Chen\thanks{State Key Laboratory for Novel Software Technology, Nanjing University, 163 Xianlin Avenue, Nanjing, Jiangsu Province, China. \textnormal{E-mails: \url{chenxiaoyu233@smail.nju.edu.cn}, \url{zhangxy@smail.nju.edu.cn}, \url{zou.zongrui@smail.nju.edu.cn}}}
  \and
  Heng Guo\thanks{School of Informatics, Edinburgh, EH8 9AB, UK. \textnormal{E-mail: \url{hguo@inf.ed.ac.uk}}. This project has received funding from the European Research Council (ERC) under the European Union’s Horizon 2020
  research and innovation programme (grant agreement No. 947778).}
\and 
Xinyuan Zhang\footnotemark[1]
\and 
Zongrui Zou\footnotemark[1]
}
\begin{document}

\maketitle

\begin{abstract}
  We give a $\widetilde{O}(n)$ time almost uniform sampler for independent sets of a matroid, whose ground set has $n$ elements and is given by an independence oracle. 
  As a consequence, one can sample connected spanning subgraphs of a given graph $G=(V,E)$ in $\widetilde{O}(\abs{E})$ time.
  This leads to improved running time on estimating all-terminal network reliability.
  Furthermore, we generalise this near-linear time sampler to the random cluster model with $q\le 1$.
\end{abstract}



\section{Introduction}
Let $\+M = ([n], \+I)$ be a matroid of rank $r$ and $\*\lambda \in \^R^n_{> 0}$ be the external fields (namely weights for the ground set elements).
Denote its set of bases by $\+B=\+B(\+M)$, and by $\+I=\+I(\+M)$ the set of independent sets.
Suppose that we want to sample a random bases $B\in\+B$ from the following distribution:
\begin{align*}
  \forall B \in \+B, \quad \muB(B) \propto \prod_{i \in B} \lambda_i.
\end{align*}
There is a natural Markov chain, namely the bases-exchange walk (also known as the down-up walk) \cite{FM92}, that converges to the distribution above.
Anari, Liu, Oveis Gharan, and Vinzant \cite{anari2019logconcaveII} showed that this chain mixes in polynomial time.
Subsequently, Cryan, Guo, and Mousa \cite{CGM21} and a follow up work by Anari, Liu, Oveis Gharan, Vinzant, and Vuong \cite{anari2021logconcave} refined the mixing time to the optimal $O(r\log r)$.

In this work, we focus on another important distribution associated with the matroid $\+M$, namely, the distribution $\muM$ over the independent sets of $\+M$: 
\begin{align} \label{eq:muM}
  \forall S \in \+I, \quad \muM(S) \propto \prod_{i \in S} \lambda_i.
\end{align}

As suggested by the previous work~\cite{anari2021logconcave}, in order to sample from $\muM$, 
we may construct another matroid $\+\+M_{\+I}$ so that there is a one-to-one correspondence between the bases of $\+M_{\+I}$ and the independent sets of $\+M$. 
Therefore, we may use the bases-exchange walk on $\+M_{\+I}$ to approximately generate samples from $\mu_{\+M,\*\lambda}$ within $O(n \log n)$ steps. 

However, an efficient implementation of the bases-exchange walk on $\+M_{\+I}$ is far from trivial.
Given an independence oracle $\+O_I$, the na{\"i}ve implementation requires $O(n)$ queries each step.
This prevents us from getting a near-linear time sampler for many potential applications.

To get faster algorithms, another oracle $\+O'$ is considered in~\cite{anari2021logconcave}.
$\+O'$ takes a set $S \subseteq [n]$ as input which contains at most one circuit.
If the circuit exists, $\+O'$ will output a uniformly random element on the circuit.
\cite{anari2021logconcave} showed that there is a sampling algorithm for $\muM$ with $O(n \log n)$ queries of oracle $\+O'$.
In the general case, the oracle $\+O'$ could be implemented by $O(r)$ calls of the independent oracle $\+O_I$, where $r$ is the rank of the matroid $\+M$.
This leads to an $O(rn \log n)$ time algorithm when $\+O_I$ answers each query in $O(1)$ time.

In this work, we give an $O(n \log^2 n)$ time sampler assuming that $\+O_I$ is equipped and it answers each query in $O(1)$ time.
This surpasses the current best $O(r n \log n)$ running time~\cite{anari2021logconcave} when $r = \Omega(\log n)$.
The metric between two distributions $\mu$ and $\nu$ over a finite space $\Omega$ we use is the \emph{total variation distance (TV distance)}: $\DTV{\mu}{\nu} := \frac{1}{2} \sum_{X \in \Omega} \abs{\mu(X) - \nu(X)}$.

\begin{theorem}\label{thm:main-theorem}
  Equipped with the independence oracle $\+O_I$ of a matroid $\+M = ([n], \+I)$, there exists an algorithm that takes external fields $\*\lambda \in \mathbb{R}_{>0}^n$ and $\epsilon \in (0,1)$ as inputs, 
  and outputs a random set $S \in \+I$ satisfying $\DTV{\muM}{S} \leq \epsilon$.
  It runs in $O\tp{(1 + \lambda_{\max})n \log(n/\epsilon) (\log n + t_{\+O_I})}$ time in expectation,
  where $t_{\+O_I}$ is the time to answer a query by the independence oracle $\+O_I$ and $\lambda_{\max} := \max_{i \in [n]} \lambda_i$.
\end{theorem}
The proof of \Cref{thm:main-theorem} is given in \Cref{sec:algo}.
\begin{remark} \label{remark:data-structure}
  In particular, instead of the independence oracle $\+O_I$,
  for our algorithm in \Cref{thm:main-theorem},
  it suffices to have a data structure maintaining a set $S \subseteq [n]$ which supports:
  \begin{itemize}
  \item to \emph{insert} an element to $S$;
  \item to \emph{delete} an element from $S$;
  \item and to \emph{query} if $S \in \+I$.
  \end{itemize}
  Given such a data structure, $t_{\+O_I}$ in \Cref{thm:main-theorem} can be substituted by the worst case or amortized running time of these operations.
\end{remark}


The crux of \Cref{thm:main-theorem} is a fast implementation of the transition step of the bases-exchange chain for $\+M_{\+I}$.
Note that the transition of a Markov chain is in itself yet another sampling problem.
We design a simple rejection sampling procedure for this latter sampling task.
There is a constant upper bound for the rejection probability (see \cref{lem:rej-ub}),
guaranteeing its efficiency.

A closely related problem is to approximate the \emph{all-terminal network reliability}.
Given a connected undirected graph $G = (V, E)$ 
and failure probabilities $\*p \in \^R^E_{>0}$,
the all-terminal network reliability $\Zrel$ is the probability that 
the graph is connected if each edge $e$ fails (i.e.~is removed) independently with probability $p_e$.
Formally, for $S \subseteq E$, let
\begin{align} \label{eq:weight-S}
  \wt(S) \defeq \*1[\text{$G[E\setminus S]$ is connected}] \cdot \prod_{e \in S} p_e \prod_{f \in E\setminus S} (1 - p_f),
\end{align}
where $G[E\setminus S]$ is the spanning subgraph of $G$ on $E\setminus S$.
Then, the reliability of the network is
\begin{align*}
  \Zrel \defeq \sum_{S\subseteq E} \wt(S).
\end{align*}
By standard techniques~\cite{jerrum1986random, stefankovic2009adaptive, kolmogorov2018faster},
estimating $\Zrel$ can be reduced to approximate sampling of the (weighted) distribution of connected spanning subgraphs:
\begin{align} \label{eq:def-murel}
  \forall S \subseteq E, \quad \murel(S) \propto \wt(S).
\end{align}

The study of the computational complexity of network reliability was initiated by Valiant \cite{Val79}.
Exact evaluation of the all-terminal version is known to be \numP-hard~\cite{Jerrum81, provan1983complexity}.
Guo and Jerrum~\cite{guo2019polynomial} gave the first fully polynomial-time randomized approximate scheme (FPRAS) using the \emph{partial rejection sampling} framework \cite{GJL19}.
This algorithm samples from $\murel$ in $O(\abs{E} + \frac{p_{\max} \abs{V}\abs{E}}{1 - p_{\max}})$ time in expectation~\cite{guo2020tight} where $p_{\max} := \max_i p_i$ is the maximum failure probability,
and this bound is tight for the technique.
%
It is also worth mentioning that, using the result in~\cite{anari2021logconcave} directly, it is possible to get an $\widetilde{O}(\abs{V}\abs{E})$ time sampler, 
whose running time is of roughly the same order as the partial rejection sampling algorithm.

Using \Cref{thm:main-theorem}, we obtain an $\widetilde{O}(\abs{V})$ speed-up to sample from $\murel$.
This gives the first near-linear time sampler for connected spanning subgraphs.
\begin{corollary}\label{cor:net-rel}
  Let $G=(V, E)$ be a connected graph with $n$ vertices and $m$ edges.
  Let $\*p \in (0, 1)^E$ be the failure probabilities for edges.
  There is an algorithm that takes $G$, $\*p$ and $\epsilon \in (0, 1)$ as input, and outputs a random subset $S\subseteq E$ such that $\DTV{\murel}{S} \leq \epsilon$ in $O\tp{\frac{m(\log^3n+\log\frac{1}{\eps})}{1-p_{\max}}}$ time in expectation, where $p_{\max} := \max_{e\in E} p_e$.
\end{corollary}

\begin{proof}
  Let $\Omega := \{S \subseteq E \mid \murel(S) > 0\}$ be the support of $\murel$.
  Recall that $\murel(S) > 0$ if and only if $G[E\setminus S]$ is connected.
  This means that there is a spanning tree $T$ of $G$ contained in $E \setminus S$.
  Note that spanning trees are bases of the graphic matroid $\+M_G$ of a graph $G$.
  Hence, let $\+M_{\-{NR}} := (E, \Omega)$, it holds that $\+M_{\-{NR}}$ is the dual matroid of $\+M_G$, namely the co-graphic matroid.
  This also means that $\murel = \muM$ for $\+M = \+M_{\-{NR}}$ and $\lambda_e = \frac{p_e}{1-p_e}, \forall e \in E$ as defined in \eqref{eq:muM}.

  It remains to implement the independence oracle efficiently.
  For this we use dynamic data structures for connectivity of graphs,
  which is a topic that has been extensively studied.
  For $\+M_{\-{NR}}$, as in \Cref{remark:data-structure},
  we implement the independence oracle $\+O_I$ with amortized cost $t_{\+O_I} = O(\log^2n)$ by using the data structure in~\cite[Section 3]{wulff2013faster} directly.
  The corollary follows by combining this with \Cref{thm:main-theorem}.
\end{proof}

Using the counting to sampling reduction in \cite{guo2020tight},
\Cref{cor:net-rel} implies an FPRAS that outputs an $(1\pm\eps)$-approximation of $\Zrel$ in time $O\tp{\frac{mn\log^4(n)}{\eps^2(1-p_{\max})}\log\frac{1}{1-p_{\max}}}$.
As before, this improves the previous best running time by a factor of $\widetilde{O}(n)$.
We also note that the running time in \Cref{cor:net-rel} is linear in $(1-p_{max})^{-1}$,
whereas the na\"ive implementation of the down-up walk has logarithmic dependence.
We leave improving this dependence for near-linear time samplers as an open problem.

The distribution $\murel$ in \eqref{eq:def-murel} is a special case of the random cluster model on the graph $G$ with parameter $q=0$ \cite{FK72}.
More generally, for a matroid $\+M=([n],\+I)$ with a rank function $\rk(\cdot)$,
the random cluster model with parameter $q\ge 0$ and external fields $\*\lambda \in \^R^n_{> 0}$ is defined as follows:
for $S\subseteq X$,
\begin{align}  \label{eqn:RC}
  \pi_{RC,q}(S) \propto q^{-\rk(S)} \prod_{x_i\in S}\lambda_i.
\end{align}
For $q=0$, the support of the distribution in \eqref{eqn:RC} must have the highest rank.
For a graphic matroid over a graph $G$, this means that $G[S]$ must be connected,
namely $S$ corresponds to $E\setminus S$ in $\murel$.

We also extends our near-linear time sampler to random cluster models with $q\le 1$.
Note that here we need a rank oracle instead of the independence oracle.

\begin{theorem}  \label{thm:RC}
  Let $0\le q\le 1$ be a parameter.
  Equipped with the rank oracle $\+O_r$ of a matroid $\+M = ([n], \+I)$, there exists an algorithm that takes external fields $\*\lambda \in \mathbb{R}_{>0}^n$ and $\epsilon \in (0,1)$ as inputs, 
  and outputs a random set $S$ such that $\DTV{\pi_{RC,q}}{S} \leq \epsilon$.
  It runs in $O\tp{(1 + \lambda_{\min}^{-1})n \log(n/\epsilon) (\log n + t_{\+O_r})}$ time in expectation,
  where $t_{\+O_r}$ is the time to answer a query by the rank oracle $\+O_r$ and $\lambda_{\min} \defeq \min_{i \in [n]} \lambda_i$.
\end{theorem}

\Cref{thm:RC} is proved in \Cref{sec:RC}.

Similar to \Cref{remark:data-structure},
it suffices to replace the rank oracle by a data structure that supports insertion/deletion of elements and query if the rank changes after removing an element.
For graphs, this can be implemented using the data structure in \cite{wulff2013faster}.
This is because $\rk(S)=\abs{V}-1+\kappa(E)-\kappa(S)$, where $\kappa(S)$ is the number of connected components in $G[S]$.
Thus the rank change query in graphs is exactly the same as asking if $u$ and $v$ are connected after removing an edge $(u,v)$.
The amortized cost of using this data structure is $O(\log^2 n)$.

\begin{corollary}\label{cor:RC-graph}
  Let $G=(V, E)$ be a graph with $n$ vertices and $m$ edges.
  Let $q\le 1$ be a parameter and $\*\lambda \in \mathbb{R}_{>0}^E$ be the external fields on edges.
  There is an algorithm that takes $G$, $\*\lambda$ and $\epsilon \in (0, 1)$ as input, and outputs a random subset $S\subseteq E$ such that $\DTV{\pi_{RC,q}}{S} \leq \epsilon$ in $O\tp{(1 + \lambda_{\min}^{-1})m (\log^3n+\log\frac{1}{\eps})}$ time in expectation,
  where $\lambda_{\min} \defeq \min_{i \in [n]} \lambda_i$.
\end{corollary}

\section{Preliminaries}
\subsection{Matroid}
Matroid is an abstract combinatorial structure that generalizes the notion of linear independence.
It is usually specified by a pair $\+M = (U, \+I)$ where $U$ is a ground set and $\+I \subseteq 2^U$ is a collection of subsets of $U$.
The subsets in $\+I$ are known as the independent sets of the matroid and satisfy the following axioms:
\begin{itemize}
\item $\emptyset \in \+I$;
\item if $S \in \+I$, $T \subseteq S$, then $T \in \+I$;
\item if $S, T \in \+I$ and $\abs{S} > \abs{T}$, then there is an element $i \in S\setminus T$ such that $T \cup \{i\} \in \+I$.
\end{itemize}
The first axiom ensures that $\+I$ is non-empty.
The second shows that $\+I$ is downward closed, and the third implies that the cardinality of maximal independent sets are the same.
This maximum cardinality is known as the \emph{rank} of the matroid $\+M$.
The set of \emph{bases} $\+B = \+B(\+M)$ is the collection of independent sets of maximum cardinality.
The rank also extends as a function to subsets of $U$.
For $S\subseteq U$, $\rk(S)$ is defined as the size of the maximum independent sets contained in $S$.

Given a matroid $\+M = (U, \+I)$, its \emph{dual matroid} $M^\star = (U, \+I^\star)$ has the same ground set $U$ with the collection of independent sets $\+I^\star := \{S \subseteq U \mid \exists B \in \+B(\+M), B\subseteq U\setminus S\}$.
By definition, every base $B^\star$ of $M^\star$ is the complement of the base $B = U \setminus B^\star$ of $\+M$ and vice versa.



%

\subsection{Strongly log-concave polynomial}
Let $f \in \^R[x_1,x_2,\ldots,x_n]$ be a polynomial with non-negative coefficients. 
$f$ is called 
\begin{itemize}
  \item \emph{$r$-homogeneous} if the degree of every monomial in $f$ is $r$;
  \item \emph{multiaffine} if every variable appears with degree no more than $1$;
  \item \emph{log-concave over the first orthant} (or \emph{log-concave} for short) if $\log f$ is concave over $\^R_{> 0}^n$, i.e., for $x, y \in \^R^n_{>0}$ and $\lambda \in (0, 1)$,
  \begin{align*}
    f(\lambda x + (1 - \lambda)y) \geq f(x)^\lambda f(y)^{1-\lambda};
  \end{align*}
  \item \emph{strongly log-concave} if $f$ is either vanishes or log-concave after taking any sequence of partial derivatives.
\end{itemize}

The notion of strong log-concavity is introduced by Gurvitz \cite{Gur09a,Gur09b}.
For a homogeneous polynomial,
it turns out to be equivalent to related notions of \emph{complete log-concavity} by Anari, Oveis Gharan, and Vinzant \cite{anari2018logconcaveI} and \emph{Lorentzian} by Br{\"a}nd{\'e}n and Huh \cite{branden2020lorentzian}.
See~\cite[Theorem 2.30]{branden2020lorentzian}.
For simplicity, we will not define the latter two.

A well known fact is that affine transform $T: \^R^m \to \^R^n$ (i.e., $T(\*y) = A\*y + \*b$ for some $A \in \^R^{n \times m}$ and $\*b \in \^R^n$) preserves log-concavity.
\begin{lemma}[\text{\cite[Lemma 2.1]{anari2018logconcaveI}}] \label{lem:affine-lc}
  If $f \in \^R[x_1, \cdots, x_n]$ is log-concave and $T: R^m \to R^n$ is an affine transformation such that $T(R^m_{>0}) \subseteq R^n_{>0}$, then $f(T(y_1, \cdots, y_m)) \in \^R[y_1, \cdots y_m]$ is log-concave.
\end{lemma}

As observed in~\cite{anari2021logconcave}, for a multiaffine polynomial $f$, its partial derivative is given by 
\begin{align*}
  \partial_1 f = \lim_{c\to\infty} \frac{f(c, x_2, \cdots, x_n)}{c},
\end{align*}
which means the derivatives of a multiaffine log-concave polynomial are limits of log-concave polynomials, which are also log-concave by definition.
It implies that a multiaffine log-concave polynomial is automatically strongly log-concave.
\begin{fact}[\cite{anari2021logconcave}]
  If polynomial $f$ is multiaffine and log-concave, then $f$ is strongly log-concave.
\end{fact}
We note that most polynomials we consider are multiaffine, which means that log-concavity and strong log-concavity are equivalent within the scope of this work.


\subsection{Polynomial and distribution}

Let $\mu$ be a distribution over $2^{[n]}$. The generating function of $\mu$ is given by $g_{\mu} = \sum_{S \subseteq [n]} \mu(S) \prod_{i \in S} x_i$.

It is known that multiaffine polynomials are closely related to generating polynomials of distribution. 
Let $f$ be a multiaffine polynomial $f \in \mathbb{R}[x_1,x_2,\ldots,x_n]$ with non-negative coefficients. 
If $f \neq 0$, there exists a distribution $\mu$ over $2^{[n]}$ such that its generating polynomial is identical to $f$ up to a scaling factor.
Hence, We may also say $f$ is the generating polynomial of $\mu$.

Furthermore, if the generating function $g_{\mu}$ of $\mu$ is $r$-homogeneous and log-concave, then the support of $\mu$ must be the set of bases of a matroid~\cite{branden2020lorentzian}.


\subsection{Down-up walk}

Let $\mu$ be a distribution over $\binom{[n]}{r}$. 
Let $\Omega(\mu) := \set{S \subseteq [n] \mid \mu(S) > 0}$ be the support of $\mu$.
A classical method for sampling from this homogeneous distribution is the down-up walk, described below.
\begin{definition} \label{def:down-up-walk}
  For a distribution $\mu$ over $\binom{[n]}{r}$, the \emph{down-up walk} $P$ updates a configuration $S \in \binom{[n]}{r}$ according to the following rule:
  \begin{enumerate}
  \item select a subset $T\subseteq S$ of size $r-1$ uniformly at random;
  \item update $S$ to $S'$ by selecting $S' \supseteq T$ with probability proportional to $\mu(S')$.
  \end{enumerate}
\end{definition}
When the support of $\mu$ is the set of bases of a matroid, 
this walk is also known as the \emph{bases-exchange walk}.

If the down-up walk $P$ connects $\Omega(\mu)$,
then $\mu$ is its unique stationary distribution.
Its \emph{mixing time} is defined by 
\begin{align*}
  t_{\-{mix}}(\epsilon) := \min \set{t \left\vert \max_{S \in \Omega(\mu)} \DTV{P^t(S, \cdot)}{\mu} \leq \epsilon \right. }.
\end{align*}
The down-up walk mixes rapidly if $g_\mu$ is (strongly) log-concave~\cite{CGM21, anari2021logconcave}.

\begin{proposition}[\text{\cite[Theorem 1]{anari2021logconcave}}]\label{prop:slc-mixing-time}
  If $g_\mu$ is $r$-homogeneous and log-concave, the mixing time of the down-up walk can be bounded by $t_{\-{mix}}(\epsilon) = O(r \log (r/\epsilon))$.
\end{proposition}

\section{Our algorithm}
\label{sec:algo}

In this section, we prove \Cref{thm:main-theorem}.
Our main tool is the down-up walk in \Cref{def:down-up-walk}.
As the uniform distribution over independent sets is not homogeneous,
in \Cref{sec:mixing} we first consider a standard homogenization,
namely its polarized version.
Then standard results imply that the down-up walk for the polarized homogeneous distribution mixes in time $O(n\log n)$.
In \Cref{sec:implement}, we show how to implement the down-up walk with $\widetilde{O}(1)$ cost.
With these ingredients, the proof of \Cref{thm:main-theorem} is given at the end of this section.

\subsection{Down-up walk for polarized polynomial}
\label{sec:mixing}
%
Let $\+M = ([n], \+I)$ be a rank-$r$ matroid and $\*\lambda \in \^R^n_{>0}$ be the external fields.
Consider
\begin{align*} 
  g(x_1, \cdots, x_n) := \sum_{S\in \+I} \prod_{i \in S} x_i.
\end{align*}
It is straightforward to verify that $g(\lambda_1 x_1, \cdots, \lambda_n x_n)$ is the generating polynomial of $\muM$.

Note that $g$ is not homogeneous, which means that we may not directly employ the down-up walk to sample from the distribution $\mu_{\+M,\*\lambda}$.
However, there is a homogeneous variant of $g$,
\begin{align} \label{eq:gh}
  g_h(y, x_1, \cdots, x_n) := \sum_{S \in \+I} y^{n - \abs{S}} \prod_{i \in S} x_i.
\end{align}
As a key step in the proof of Mason's ultra-log-concavity conjecture for independent sets of matroid~\cite{anari2018logconcave, branden2020lorentzian}, the following result is proved.
\begin{lemma}[\text{\cite[Theorem 4.1]{anari2018logconcave}}]\label{lem:slc-gh}
  The polynomial $g_h$ in \eqref{eq:gh} is strongly log-concave. 
\end{lemma}

However, $g_h$ is not multiaffine, which means that it is not a generating polynomial of any distribution.
Instead, we consider the polarized version of $g_h$.
\begin{lemma}
  \label{lem:slc-gp}
  If $g_h$ in \eqref{eq:gh} is strongly log-concave, then the following polynomial is also strongly log-concave:
  \begin{align}\label{eq:gp}
    g_{p}(x_1,\cdots,x_n,y_1,\cdots,y_n) = \sum_{S \in \+I} \frac{e_{n-\abs{S}}(\*y)}{\binom{n}{\abs{S}}}  \prod_{i \in S} x_i,
  \end{align}
  where $e_k(\*y) = \sum_{1 \le i_1 < i_2< \ldots<i_k \le n} \prod_{j=1}^k y_{i_j}$ is the $k$-th elementary symmetric polynomial.
\end{lemma}
We remark that \Cref{lem:slc-gp} is a special case of \cite[Proposition 3.1]{branden2020lorentzian},
and it is more explicitly derived in \cite[Section 6.6]{Mousa22}.

The polarized polynomial in \eqref{eq:gp} corresponds back to a distribution.
Let $X := \{x_1, \cdots, x_n\}$ denote elements in $\+M$ and $Y := \{y_1, \cdots, y_n\}$ denote the auxiliary variables introduced by polarization.
Let $\pi$ be the distribution over subsets of $X \cup Y$ corresponding to the generating polynomial $g_p(\lambda_1 x_1,\ldots,\lambda_n x_n, y_1,\ldots,y_n)$.
Then the support of $\pi$ is given by
\begin{align*}
  \Omega(\pi) = \{A \cup B \mid A \in \+I, B \subseteq Y, \abs{A}+\abs{B} = n\}.
\end{align*}
Furthermore, for every $S = A \cup B \in \Omega(\pi)$, it holds that 
\begin{align*}
  \pi(S) &\propto \frac{1}{\binom{n}{\abs{A}}} \prod_{x_i \in A} \lambda_i,
\end{align*}
and $\sum_{S: S\cap X=A}\pi(S) = \muM(A)$.

Therefore, to sample from $\muM$ within a TV distance of $\epsilon$, it suffices to sample $S \in \Omega(\pi)$ such that $\DTV{S}{\pi} \leq \epsilon$, and then return $S \cap X$ as the result. 


Note that $g_p$ is homogeneous and multiaffine.
Moreover, according to \Cref{lem:slc-gh} and \Cref{lem:slc-gp}, $g_p$ is strongly log-concave.
By \Cref{lem:affine-lc}, the polynomial $g_p(\lambda_1 x_1, \cdots, \lambda_n x_n, y_1, \cdots, y_n)$ is also log-concave.
Hence, by \Cref{prop:slc-mixing-time}, we have the following result, which gives a powerful framework to build fast sampling algorithm for independent set of matroid.
\begin{lemma} \label{lem:mix-gp}
  The down-up walk $P$ of $\pi$ mixes in time $O(n\log(n/\epsilon))$.
\end{lemma}
%
We note that the mixing time $O(n\log(n/\epsilon))$ does not readily imply a sampler with $O(n \log (n/\epsilon))$ running time, 
as it may take $\omega(1)$ time to implement a single transition step of $P$.
According to \Cref{def:down-up-walk}, in each step, $P$ updates a state $S \in \Omega(\pi)$ as in \Cref{alg:down-up-walk-pi}.
\begin{algorithm}
  \caption{a step of down-up walk $P$ on $\pi$}\label{alg:down-up-walk-pi}
  select a subset $T \subseteq S$ of size $n-1$ uniformly at random; \\
  update $S$ to $S'$ by selecting random $S' \supseteq T$ according to the following law:
    \begin{align} \label{eq:target-prob}
      \Pr{S'} \propto
      \*1[S' \cap X \in \+I] \times
      \begin{cases}
        \frac{1}{\binom{n}{\abs{T\cap X}}}, & S' \setminus T \in Y; \\
        \frac{\lambda_i}{\binom{n}{\abs{T\cap X} + 1}}, & S'\setminus T \in X.
      \end{cases}
    \end{align}
\end{algorithm}

The main obstacle is to implement the second step in \Cref{alg:down-up-walk-pi} efficiently.
A na{\"i}ve approach checks whether $(\{x_i\} \cup T)\cap X \in \+I$ for each $x_i \in X \setminus T$  by calling the independence oracle $\+O_I$, 
and then generates a random sample from all ``feasible'' $x_i$ and together with all of $y_i \in Y \setminus T$ according to the desired distribution in \eqref{eq:target-prob}.
In the worst case, this gives an $O(n)$ overhead and the running time of the sampling algorithm becomes $O(n^2 \log n)$.

\subsection{A fast implementation of the down-up walk}
\label{sec:implement}

Our main contribution is an efficient implementation of the down-up walk $P$ on $\pi$, where each step of $P$ takes constant time in expectation given an independence oracle $\+O_I$.
In fact, the implementation task is yet another sampling problem from the distribution in \eqref{eq:target-prob},
and we do so by rejection sampling, described in \Cref{alg:rej-sampling}.


\begin{algorithm}[H]
  \caption{implementation for the second step of the down-up walk $P$}\label{alg:rej-sampling}
  \SetKwInOut{Input}{input}
  \SetKwInOut{Output}{output}
  \SetKwInOut{Parameter}{parameter}
  \Input{a subset $T \subseteq X\cup Y$ of size $n-1$ such that $T\cap X \in \+I$}
  \Output{a random configuration $S$ according to the distribution defined in \eqref{eq:target-prob}}
  \While{true}{
    \label{line:samping}
    propose an element $e \in (X\cup Y) \setminus T$ according to the following distribution $\nu$:
    \begin{align} \label{eq:nu}
      \forall e \in (X \cup Y) \setminus T, \quad \nu(e) \propto 
      \begin{cases}
        \lambda_i, & e = x_i \in X \setminus T;\\
        \frac{n - \abs{T\cap X}}{1 + \abs{T\cap X}}, & e \in Y \setminus T.
      \end{cases}
    \end{align}\\
    \If{$(T \cup \{e\}) \cap X \in \+I$}{ \label{line:reject}
      \textbf{return} $S = T\cup \{e\}$;
    }
  }
\end{algorithm}

The correctness of \Cref{alg:rej-sampling} is straightforward.
\begin{fact} \label{fact:correct-rej-sampling}
  The state $S$ produced by \Cref{alg:rej-sampling} follows the distribution defined in \eqref{eq:target-prob}.
\end{fact}

In terms of efficiency, the \textbf{while} loop in \Cref{alg:rej-sampling} is anticipated to execute for a constant number of rounds in expectation.
This is because the rejection probability is upper bounded by a constant, as shown by the next lemma.
\begin{lemma} \label{lem:rej-ub}
  It holds that $\Pr[e \sim \nu]{(T \cup \{e\})\cap X \not\in \+I} \leq \frac{\lambda_{\max}}{1 + \lambda_{\max}}$, where $\lambda_{\max} = \max_{i \in [n]} \lambda_i$.
\end{lemma}
\begin{proof}
  Suppose $\abs{T\cap X} = k$.
  Note that if $e \in Y\setminus T$, then $(T\cup \{e\})\cap X \in \+I$.
  This means that
  \begin{align*}
    \Pr[e\sim \nu]{(T\cup \{e\})\cap X \not\in \+I} & \le \sum_{x_i\in X\setminus T}\frac{\lambda_i}{\sum_{x_i\in X\setminus T}\lambda_i+\sum_{y\in Y\setminus T}\frac{n-k}{1+k}}\\
    & = \frac{\sum_{x_i\in X\setminus T}\lambda_i}{\sum_{x_i\in X\setminus T}(1+\lambda_i)} \le \frac{\lambda_{\max}}{1 + \lambda_{\max}},
  \end{align*}
  where the equality is due to $\abs{Y\cap T}+k=n-1$.
%
\end{proof}

Now, we are ready to prove \Cref{thm:main-theorem}.
\begin{proof}[Proof of \Cref{thm:main-theorem}]
  Our algorithm is just running \Cref{alg:down-up-walk-pi} for $O(n\log(n/\epsilon))$ steps and then output $S\cap X$.
  Line 2 of \Cref{alg:down-up-walk-pi} is implemented by \Cref{alg:rej-sampling}, to get a random state $S$.
  By \Cref{lem:mix-gp} and \Cref{fact:correct-rej-sampling}, it holds that $\DTV{S\cap X}{\mu} \leq \epsilon$.

  To implement \eqref{eq:nu}, we maintain two balanced binary search trees $T_X$ and $T_Y$ that keep track of the weight of each node and the sum of weights in each subtree.
  The first tree $T_X$ maintains elements $x_i \in X\setminus T$ each assigned with weight $\lambda_i$,
  and the second tree $T_Y$ maintains elements $y_i \in Y\setminus T$ with weight $1$.
  
  To produce an $e \sim \nu$, first choose tree $T_Z \in \{T_X, T_Y\}$ randomly according to the following law:
  \begin{align*}
    T_Z =
    \begin{cases}
      T_X & \text{with prob.} \propto \sum_{x_i \in X\setminus T} \lambda_i \\
      T_Y & \text{with prob.} \propto n - \abs{T\cap X}
    \end{cases},
  \end{align*}
  where we note that $\sum_{x_i \in X\setminus T} \lambda_i$, and $\abs{T\cap X}$ could be obtained by a constant number of queries via the binary search trees.
  To sample an element $e \in T_Z$ according to the weights of each element, we may consider a binary search algorithm on $T_Z$ that runs in $O(\log n)$ time. 
  We initialize a variable $e$ with the root of $T_Z$, and then repeat the following procedure:
  \begin{enumerate}
    \item Let $L$ be the sum of weights in the left subtrees of $e$, $R$ be the sum of weights in the right subtrees, and $W$ be the weights of $e$;
    \item Sample a real number $0 < x < L+R+W$ uniformly at random. If $x < W$, return $e$; else if $x < W+L$, update $e$ to the left child of $e$; otherwise, update $e$ to the right child of $e$.
  \end{enumerate}
  Finally, by \Cref{lem:rej-ub}, the rejection sampling procedure in \Cref{alg:rej-sampling} runs within $O(1 + \lambda_{\max})$ rounds in expectation.
  Also note that in each round of the rejection sampling, we need $t_{\+O_I}$ time to query the independence oracle.
  Together, the algorithm runs in
  \[O((1 + \lambda_{\max}) n \log (n/\epsilon) (\log n + t_{\+O_I}))\]
  time in expectation.
\end{proof}

\section{Random cluster models with \texorpdfstring{$q\le 1$}{q<=1}}
\label{sec:RC}

Once again, let $\+M=(X,\+I)$ be a matroid, equipped with a rank function $\rk(\cdot)$.
For $i\in[n]$, let $\lambda_i>0$ be the weight or external field of $x_i\in X$.
Let $0\le q\le 1$ be a parameter.
Recall the definition of the random cluster model \cite{FK72} in \eqref{eqn:RC}.
Note that when $q=0$, the support of $\pi_{RC}$ are all subsets of full rank,
namely they are the complements of the independent sets of the dual matroid.

Similar to \Cref{sec:mixing},
as the distribution is not homogeneous, we want to polarize it.
Let $Y$ be a set of $n=\abs{X}$ auxiliary variables.
For $T\subseteq X\cup Y$ such that $\abs{T}=n$,
the polarized distribution
\begin{align}  \label{eqn:RC-polar}
  \RCpolar(T) \propto \frac{q^{-\rk(T\cap X)} \prod_{x_i\in T\cap X}\lambda_i}{\binom{n}{\abs{T\cap Y}}}.
\end{align}
Let the right hand side of \eqref{eqn:RC-polar} be $\wt(T)$.
Note that the marginal distribution of $\RCpolar$ on $X$ is the same as $\pi_{RC,q}$.

Once homogenized, we may consider the up-down walk for $\RCpolar$.
For the up step, we uniformly add an element from $(X\cup Y) \setminus T$.
For the down step, suppose the current set is $T$ such that $\abs{T}=n+1$.
We want to remove an element $e\in T$ with probability proportional to $\wt(T\setminus\{e\})$.
Namely, the transition probability $p(e)$ satisfies
\begin{align}  \label{eqn:up-prob}
  p(e) \propto 
  \begin{cases}
    \frac{q^{-\rk(T\cap X)} \prod_{x_i\in T\cap X}\lambda_i}{\binom{n}{\abs{T\cap Y}-1}} & \text{if $e\in T\cap Y$;}\\
    \frac{q^{-\rk(T\cap X)} \prod_{x_i\in T\cap X}\lambda_i}{\lambda_j\binom{n}{\abs{T\cap Y}}} & \text{if $e=x_j\in T\cap X$ and $\rk(T\cap X\setminus e)=\rk(T\cap X)$;}\\
    \frac{q^{-\rk(T\cap X)+1} \prod_{x_i\in T\cap X}\lambda_i}{\lambda_j\binom{n}{\abs{T\cap Y}}} & \text{if $e=x_j\in T\cap X$ and $\rk(T\cap X\setminus e)=\rk(T\cap X)-1$.}
  \end{cases}
\end{align}
We may further normalize it to get
\begin{align}  \label{eqn:up-prob-normalize}
  p(e) \propto 
  \begin{cases}
    \frac{n-\abs{T\cap Y}+1}{\abs{T\cap Y}} = \frac{\abs{T\cap X}}{\abs{T\cap Y}} & \text{if $e\in T\cap Y$;}\\
    \lambda_j^{-1} & \text{if $e=x_j\in T\cap X$ and $\rk(T\cap X\setminus e)=\rk(T\cap X)$;}\\
    q\lambda_j^{-1} & \text{if $e=x_j\in T\cap X$ and $\rk(T\cap X\setminus e)=\rk(T\cap X)-1$.}
  \end{cases}
\end{align}
To implement \eqref{eqn:up-prob-normalize},
we may first propose $e\sim \nu$ where 
\begin{align}  \label{eqn:nu}
  \nu(e) \propto
  \begin{cases}
    \frac{\abs{T\cap X}}{\abs{T\cap Y}} & \text{if $e\in T\cap Y$;}\\
    \lambda_j^{-1} & \text{if $e=x_j\in T\cap X$,}
  \end{cases}
\end{align}
and then reject $e\in T\cap X$ with probability $1-q$ if $\rk(T\cap X\setminus \{e\})=\rk(T\cap X)-1$.
Keep doing this until we accept.
To see the efficiency of this implementation.
Let $\+E$ be the event that rejection happens.
Then, 
\begin{align*}
  \Pr{\+E}&\le \sum_{x_j\in T\cap X}\frac{(1-q)\lambda_j^{-1}}{\sum_{x_j\in T\cap X}\lambda_j^{-1}+\sum_{e\in T\cap Y}\frac{\abs{T\cap X}}{\abs{T\cap Y}}}\\
  & =(1-q) \frac{\sum_{x_j\in T\cap X}\lambda_j^{-1}}{\sum_{x_j\in T\cap X}(\lambda_j^{-1}+1)} \le \frac{1-q}{1+\lambda_{\min}},
\end{align*}
where $\lambda_{\min}\defeq \min_{i \in [n]} \lambda_i$.
Thus, in expectation, we will successfully make a transition after $O(1+\lambda_{\min}^{-1})$ steps.

The mixing time of the up-down walk can be bounded in the same way as before.
The homogenized generating polynomial for random cluster models with $q\le 1$ is shown to be strongly log-concave by Br{\"a}nd{\'e}n and Huh \cite{BH18}.
Also note that the up-down walk is just the down-up walk of the dual.
By \Cref{prop:slc-mixing-time}, the mixing time of the up-down walk is $O(n\log (n/\eps))$ as well.

However, there is one difference here that we need to check the rank of $T\cap X$ and $ T\cap X\setminus \{e\}$.
Thus we need a rank oracle instead of the independence oracle.
Putting everything together, we get \Cref{thm:RC}.

\bibliographystyle{alpha}
\bibliography{refs}

\newcommand{\etalchar}[1]{$^{#1}$}
\begin{thebibliography}{ALO{\etalchar{+}}21}

\bibitem[ALO{\etalchar{+}}21]{anari2021logconcave}
Nima Anari, Kuikui Liu, Shayan {Oveis Gharan}, Cynthia Vinzant, and
  Thuy{-}Duong Vuong.
\newblock Log-concave polynomials {IV:} approximate exchange, tight mixing
  times, and near-optimal sampling of forests.
\newblock In {\em STOC}, pages 408--420. {ACM}, 2021.
\newblock arXiv:2004.07220 v2.

\bibitem[ALOV18]{anari2018logconcave}
Nima Anari, Kuikui Liu, Shayan {Oveis Gharan}, and Cynthia Vinzant.
\newblock Log-concave polynomials {III:} {M}ason's ultra-log-concavity
  conjecture for independent sets of matroids.
\newblock {\em arXiv}, abs/1811.01600, 2018.

\bibitem[ALOV19]{anari2019logconcaveII}
Nima Anari, Kuikui Liu, Shayan {Oveis Gharan}, and Cynthia Vinzant.
\newblock Log-concave polynomials {II:} high-dimensional walks and an {FPRAS}
  for counting bases of a matroid.
\newblock In {\em STOC}, pages 1--12. {ACM}, 2019.

\bibitem[AOV18]{anari2018logconcaveI}
Nima Anari, Shayan {Oveis Gharan}, and Cynthia Vinzant.
\newblock Log-concave polynomials, entropy, and a deterministic approximation
  algorithm for counting bases of matroids.
\newblock In {\em FOCS}, pages 35--46. IEEE, 2018.
\newblock arXiv:1807.00929 v2.

\bibitem[BH18]{BH18}
Petter Br{\"a}nd{\'e}n and June Huh.
\newblock {Hodge}-{Riemann} relations for {Potts} model partition functions.
\newblock {\em arXiv}, abs/1811.01696, 2018.

\bibitem[BH20]{branden2020lorentzian}
Petter Br{\"a}nd{\'e}n and June Huh.
\newblock Lorentzian polynomials.
\newblock {\em Ann. of Math. (2)}, 192(3):821--891, 2020.
\newblock arXiv:1902.03719 v6.

\bibitem[CGM21]{CGM21}
Mary Cryan, Heng Guo, and Giorgos Mousa.
\newblock Modified log-{S}obolev inequalities for strongly log-concave
  distributions.
\newblock {\em Ann. Probab.}, 49(1):506--525, 2021.

\bibitem[FK72]{FK72}
Cornelius~M. Fortuin and Pieter~W. Kasteleyn.
\newblock On the random-cluster model. {I}. {I}ntroduction and relation to
  other models.
\newblock {\em Physica}, 57:536--564, 1972.

\bibitem[FM92]{FM92}
Tom{\'{a}}s Feder and Milena Mihail.
\newblock Balanced matroids.
\newblock In {\em {STOC}}, pages 26--38. {ACM}, 1992.

\bibitem[GH20]{guo2020tight}
Heng Guo and Kun He.
\newblock Tight bounds for popping algorithms.
\newblock {\em Random Structures Algorithms}, 57(2):371--392, 2020.

\bibitem[GJ19]{guo2019polynomial}
Heng Guo and Mark Jerrum.
\newblock A polynomial-time approximation algorithm for all-terminal network
  reliability.
\newblock {\em SIAM J. Comput.}, 48(3):964--978, 2019.

\bibitem[GJL19]{GJL19}
Heng Guo, Mark Jerrum, and Jingcheng Liu.
\newblock Uniform sampling through the {L}ov{\'{a}}sz local lemma.
\newblock {\em J. {ACM}}, 66(3):18:1--18:31, 2019.

\bibitem[Gur09a]{Gur09b}
Leonid Gurvits.
\newblock On multivariate {N}ewton-like inequalities.
\newblock In {\em Advances in combinatorial mathematics}, pages 61--78.
  Springer, Berlin, 2009.

\bibitem[Gur09b]{Gur09a}
Leonid Gurvits.
\newblock A polynomial-time algorithm to approximate the mixed volume within a
  simply exponential factor.
\newblock {\em Discrete Comput. Geom.}, 41(4):533--555, 2009.

\bibitem[Jer81]{Jerrum81}
Mark Jerrum.
\newblock {\em On the complexity of evaluating multivariate polynomials}.
\newblock PhD thesis, The University of Edinburgh, 1981.

\bibitem[JVV86]{jerrum1986random}
Mark~R. Jerrum, Leslie~G. Valiant, and Vijay~V. Vazirani.
\newblock Random generation of combinatorial structures from a uniform
  distribution.
\newblock {\em Theoret. Comput. Sci.}, 43(2-3):169--188, 1986.

\bibitem[Kol18]{kolmogorov2018faster}
Vladimir Kolmogorov.
\newblock A faster approximation algorithm for the gibbs partition function.
\newblock In {\em {COLT}}, volume~75, pages 228--249. {PMLR}, 2018.

\bibitem[Mou22]{Mousa22}
Giorgos Mousa.
\newblock {\em Local-to-Global Functional Inequalities in Simplicial
  Complexes}.
\newblock PhD thesis, The University of Edinburgh, 2022.

\bibitem[PB83]{provan1983complexity}
J.~Scott Provan and Michael~O. Ball.
\newblock The complexity of counting cuts and of computing the probability that
  a graph is connected.
\newblock {\em SIAM J. Comput.}, 12(4):777--788, 1983.

\bibitem[{\v S}VV09]{stefankovic2009adaptive}
Daniel {\v S}tefankovi{\v c}, Santosh Vempala, and Eric Vigoda.
\newblock Adaptive simulated annealing: a near-optimal connection between
  sampling and counting.
\newblock {\em J. ACM}, 56(3):Art. 18, 36, 2009.

\bibitem[Val79]{Val79}
Leslie~G. Valiant.
\newblock The complexity of enumeration and reliability problems.
\newblock {\em {SIAM} J. Comput.}, 8(3):410--421, 1979.

\bibitem[WN13]{wulff2013faster}
Christian Wulff-Nilsen.
\newblock Faster deterministic fully-dynamic graph connectivity.
\newblock In {\em SODA}, pages 1757--1769. SIAM, 2013.

\end{thebibliography}

\end{document}